\documentclass[11pt]{article}

\usepackage{amsfonts}
\usepackage{amssymb}
\usepackage{amstext}
\usepackage{amsmath}
\usepackage{enumerate}
\usepackage{algorithm}
\usepackage{algorithmic}

\usepackage{amsthm}

\usepackage[usenames,dvipsnames,svgnames,table]{xcolor}

\usepackage[pagebackref,letterpaper=true,colorlinks=true,pdfpagemode=UseNone,citecolor=OliveGreen,linkcolor=BrickRed,urlcolor=BrickRed,pdfstartview=FitH]{hyperref}


\newtheorem{lemma}{Lemma}
\newtheorem{theorem}{Theorem}

\newtheorem{corollary}{Corollary}

\newtheorem{claim}{Claim}

\numberwithin{equation}{section}
\numberwithin{table}{section}

\newcommand{\junk}[1]{}

\newcommand{\ol}{\overline}

\renewcommand{\l}{\lambda}

\def\b1{{\bf 1}}
\def\eps{{\epsilon}}

\def\ver{{\phi^V}}

\def\phisw{\phi_{\text{sweep}}}

\def\polylog{\operatorname{polylog}} 
 
\def\rquo{\mathcal R} 
 
\def\expe{\mathbb E}

\DeclareMathOperator{\argmax}{argmax}

\DeclareMathOperator{\supp}{supp}

\title{Improved Cheeger's Inequality and Analysis of Local Graph Partitioning using Vertex Expansion and Expansion Profile
}

\author{
Tsz Chiu Kwok\thanks{\'Ecole polytechnique f\'ed\'erale de Lausanne, \protect\url{tckwok0@gmail.com}.}
\and
Lap Chi Lau\thanks{Simons Institute and University of Waterloo, \protect\url{lapchi@uwaterloo.ca}. This material is based upon work supported by the National Science Foundation under Grant No. 1216642.}\\
\and
Yin Tat Lee\thanks{Massachusetts Institute of Technology, \protect\url{yintat@mit.edu}.}\\
}

\date{}

\begin{document}

\begin{titlepage}
\def\thepage{}
\thispagestyle{empty}

\maketitle

\begin{abstract}
\noindent
We prove two generalizations of the Cheeger's inequality. 
The first generalization relates the second eigenvalue to the edge expansion and the vertex expansion of the graph $G$,
\[ \lambda_2 = \Omega( \ver(G) \cdot \phi(G) ),\]
where $\ver(G)$ denotes the robust vertex expansion of $G$ and $\phi(G)$ denotes the edge expansion of $G$.
The second generalization relates the second eigenvalue to the edge expansion and the expansion profile of $G$, for all $k \geq 2$, 
\[ \lambda_2 = \Omega( \frac{1}{k} \cdot \phi_k(G) \cdot \phi(G) ),\]
where $\phi_k(G)$ denotes the $k$-way expansion of $G$.
These show that the spectral partitioning algorithm has better performance guarantees when $\ver(G)$ is large (e.g. planted random instances) or $\phi_k(G)$ is large (instances with few disjoint non-expanding sets).
Both bounds are tight up to a constant factor.

Our approach is based on a method to analyze solutions of Laplacian systems, and this allows us to extend the results to local graph partitioning algorithms.
In particular, we show that our approach can be used to analyze personal pagerank vectors,
and to give a local graph partitioning algorithm for the small-set expansion problem with performance guarantees similar to the generalizations of Cheeger's inequality.
We also present a spectral approach to prove similar results for the truncated random walk algorithm. 
These show that local graph partitioning algorithms almost match the performance of the spectral partitioning algorithm, with the additional advantages that they apply to the small-set expansion problem and their running time could be sublinear.
Our techniques provide common approaches to analyze the spectral partitioning algorithm and local graph partitioning algorithms.
\end{abstract}

\end{titlepage}

\newpage

\section{Introduction}

Let $G = (V,E)$ be a complete weighted graph and $n := |V|$.
For simplicity, we assume that the graph is regular and the total weight on each vertex is one throughout\footnote{By standard arguments, the results can be extended to handle non-regular graphs using the notion of conductance and the normalized Laplacian matrix.}.
Let $w(S,T)$ be the total weight of the edges with one vertex in $S$ and another vertex in $T$.
The edge expansion of a set $S \subseteq V$ and the edge expansion of a graph $G$ are defined as
\[\phi(S) := \frac{w(S,\ol{S})}{|S|}
\quad {\rm and} \quad \phi(G) := \min_{S: |S| \leq |V|/2} \phi(S).\]
Let $L = I - A$ be the Laplacian matrix of $G$ where $I$ and $A$ are the identity and the adjacency matrix of $G$,
with eigenvalues of $L$ being $0 = \lambda_1 \leq \lambda_2 \leq \ldots \leq \lambda_n \leq 2$.
Cheeger's inequality~\cite{cheeger,alon-milman,alon} bounds the edge expansion of $G$ using the second eigenvalue of $L$, 
\[ \frac12 \lambda_2 \leq \phi(G) \leq \sqrt{2 \lambda_2}. \]
It is useful in bounding the edge expansion of a graph and also bounding the mixing time of random walks~\cite{hoory-linial-wigderson}.
The proof of Cheeger's inequality gives an efficient algorithm to find a set with expansion at most $\sqrt{2\lambda_2}$, and we will refer to this algorithm as the spectral partitioning algorithm (also known as the sweep cut algorithm on the second eigenvector).
A recent generalization~\cite{improved-cheeger} of Cheeger's inequality bounds the edge expansion of $G$ using the second and the $k$-th eigenvalues of $L$ for any $k \geq 2$,
\[ \phi(G) = O(k) \frac{\lambda_2}{\sqrt{\lambda_k}}. \]
This provides a better analysis of the spectral partitioning algorithm in practical instances of image segmentation and data clustering.

\subsection{Our results}

We prove two new generalizations of Cheeger's inequality.
These provide better analyses of the spectral partitioning algorithm when some expansion parameters of the graph are large.
We also prove similar bounds for the personal pagerank algorithm and the truncated random walk algorithm.
These give local graph partitioning algorithms for the small-set expansion problem with improved Cheeger's guarantees.
Our techniques provide common approaches to analyze the spectral partitioning algorithm and local graph partitioning algorithms.

\subsubsection{Vertex Expansion}

The first generalization bounds the second eigenvalue of $L$ by the edge expansion and the vertex expansion of $G$.
We define the robust vertex expansion following Kannan, Lov\'asz and Montenegro~\cite{kannan-lovasz-montenegro}.
For $S \subseteq V$, let $N_{1/2}(S) := \min\{ |T|~|~ T \subseteq V-S {\rm~and~} w(S,T) \geq \frac12 w(S,\ol{S})\}$.
Define 
\[\ver(S) := \frac{N_{1/2}(S)}{|S|}
\quad {\rm and} \quad \ver(G) := \min_{S:|S| \leq |V|/2} \ver(S)\]
as the robust vertex expansion\footnote{
Note that the usual definition of vertex expansion, define as $\min_{S: |S| \leq |V|/2} N(S)/|S|$, is too sensitive to edges of tiny weights (e.g. adding a complete graph with tiny edge weight will change $\ver(G)$ to one).
One could replace the constant $1/2$ in the definition of $N_{1/2}(S)$ by other constant say $0.99$ so that the definition of robust vertex expansion is closer to the definition of (ordinary) vertex expansion while we can still obtain similar results.} of $G$. 
Also define
\[\Psi(S) := \phi(S) \cdot \ver(S)
\quad {\rm and} \quad \Psi(G) := \min_{S:|S| \leq |V|/2} \Psi(S)\]
as the minimum product of the edge expansion and the robust vertex expansion.
The following is a generalization of Cheeger's inequality using robust vertex expansion.

\begin{theorem} \label{t:product}
\[\lambda_2 = \Omega(\min \{\Psi(G), \phi(G)\}).\]
\end{theorem}

\begin{corollary} \label{c:vertex}
\[\lambda_2 = \Omega(\ver(G) \cdot \phi(G)).\]
\end{corollary}

Note that $\ver(S) \geq \frac12 \phi(S)$ and so Corollary~\ref{c:vertex} is a generalization of Cheeger's inequality.
Observe that $\ver(S)$ could be much larger than $\phi(S)$ when the edges crossing $S$ spread out.
For example, randomly generated instances such as those in the planted partition model~\cite{boppana,mcsherry} have $\ver(G) = \Omega(1)$, 
and thus Theorem~\ref{t:product} implies that the spectral partitioning algorithm is a constant factor approximation algorithm for those instances\footnote{For example, in a planted $k$-partition instance where there are $k$ subsets of size $n/k$ with probability $p$ having an edge between two vertices in the same subset and probability $q$ having an edge between two vertices in different subsets for $q \ll p$, the improved Cheeger's inequality only proves a $O(k)$-approximation while Theorem~\ref{t:product} proves a $O(1)$-approximation.}.
Another interesting example is the hypercube\footnote{For hypercubes, it is known that the edge expansion is $\Omega(1/\log n)$, the vertex expansion is $\Omega(1/\sqrt{\log n})$~\cite{harper}, and the product of the edge expansion and the vertex expansion is $\Omega(1/\log n)$~\cite{margulis}.  
We believe that the same bounds hold for robust vertex expansion, $\ver(G)=\Omega(1/\sqrt{\log n})$ and $\Psi(G) = \Omega(1/\log n)$ but we don't know of a proof yet.
If that's true, Corollary~\ref{c:vertex} will give a bound of $\Omega(1/\log^{3/2}(n))$ on the second eigenvalue, and Theorem~\ref{t:product} will give the correct bound of $\Omega(1/\log(n))$, while Cheeger's inequality only gives a bound of $\Omega(1/\log^2 n)$.  
}.

\subsubsection{Expansion Profile}

The $\delta$-small-set expansion ($0 < \delta \leq 1/2$) of $G$ 
and the $k$-way expansion ($k \geq 2)$ of $G$ are defined as
\[\phi_{\delta}(G) := \min_{S: |S| \leq \delta |V|} \phi(S)
\quad {\rm and} \quad
\phi_k(G) := \min_{S_1, \ldots, S_k:~S_i \cap S_j = \emptyset~\forall i \neq j~} \max_{1 \leq i \leq k} \phi(S_i).\]
The curve $\phi_{\delta}(G)$ for $0 < \delta \leq 1/2$ is defined by Lov\'asz and Kannan~\cite{lovasz-kannan} and is called the expansion profile of $G$.
Note that $\phi(G) = \phi_{1/2}(G) = \phi_2(G)$. 
The following is a generalization of Cheeger's inequality using $k$-way expansion.

\begin{theorem} \label{t:k-way}
For all $k \geq 2$,
\[\lambda_2 = \Omega( \frac1k \cdot \phi_k(G) \cdot \phi(G) ).\]
\end{theorem}

\begin{corollary} \label{c:small-set}
For all $\delta \leq 1/2$,
\[\lambda_2 = \Omega( \delta \cdot \phi_{\delta}(G) \cdot \phi(G) ).\]
\end{corollary}

Both Theorem~\ref{t:product} and Theorem~\ref{t:k-way} are tight up to a constant factor.
Both proofs of Theorem~\ref{t:product} and Theorem~\ref{t:k-way} show that the spectral partitioning algorithm achieves the performance guarantees, i.e. the algorithm would output a set $S$ with $\Psi(S) = O(\lambda_2)$ and $\phi(S) = O(k\lambda_2/\phi_k(G))$ respectively.
These imply that the spectral partitioning algorithm is a $O(1/\ver(G))$-approximation and a $O(k/\phi_k(G))$-approximation for edge expansion.

\subsubsection{Local Partitioning Algorithms for Small-Set Expansion}

Our proof techniques allow us to use the same approach to analyze the local partitioning algorithm using personal pagerank vectors~\cite{andersen-chung-lang}.
Given a parameter $\alpha \in (0,1]$ and a vertex $s$,
the personal pagerank vector $r_{s,\alpha} \in {\mathbb R}^n$ is the unique solution to the equation $r_{s,\alpha} = \alpha \chi_s + (1 - \alpha)Wr_{s,\alpha}$, where $W$ is the transition matrix of the lazy random walks.

\begin{theorem} \label{t:pagerank}
For any (unknown target) set $S \subseteq V$, 
there is a polynomial time randomized algorithm to find a set $S'$ with 
\begin{enumerate}
\item $\phi(S') = O(\phi(S) \log(|S|) / \ver(G))$ and $|S'| = O(|S| \log |S|)$,
\item $\phi(S') = O(k \phi(S) \log(|S|) / \phi_k(G))$ and $|S'| = O(|S| \log |S|),$
\end{enumerate}
by computing $r_{s,\alpha}$ for a random vertex $s \in S$ with $\alpha = O(\phi(S))$ and returning a level set of $r_{s,\alpha}$.
For unweighted $d$-regular graphs, there is a local implementation with running time $O(d |S| \log(|S|) / \phi(S) + |S| \log^2 |S|)$.
\end{theorem}

Theorem~\ref{t:pagerank} implies that the personal pagerank algorithm is a $O(\log(|S|)/\ver(G))$-approximation and a $O(k\log(|S|) /\phi_k(G))$-approximation for the small-set expansion problem where the output set size is bounded within a logarithmic factor of the target set size.

We also present a spectral approach to prove that the local graph partitioning algorithm using truncated random walks has similar performance guarantees as the spectral partitioning algorithm.
Let $p_{s,t} := W^t \chi_s$ be the probability distribution vector after $t$ steps of lazy random walks starting from the vertex $s$.

\begin{theorem} \label{t:walks}
For any (unknown target) set $S \subseteq V$, there is a polynomial time randomized algorithm to find a set $S$ with 
\begin{enumerate}
\item $\phi(S') = O(k \phi(S) / (\eps \phi_k(G)))$ and $|S'| = O(|S|^{1+\eps})$,
\item $\phi(S') = O(k \phi(S) \log(|S|) / \phi_k(G))$ and $|S'| = O(|S|)$,
\end{enumerate}
by computing $p_{s,t}$ for a random vertex $s \in S$ 
with $t = O(\log (|S|)/\phi(S))$ for (1) 
and $t=O(1/\phi(S))$ for (2) and returning a level set of $p_{s,t}$.
For unweighted $d$-regular graphs, there is a local implementation with running time $O(d \epsilon^2 |S|^{1 + \epsilon} \log^2 (|S|) / \phi(S)^3)$.
\end{theorem}

Theorem~\ref{t:walks} implies that the truncated random walks algorithm is a $O(k/\phi_k(G))$-approximation or a $O(k \log(|S|)/\phi_k(G))$-approximation for the small-set expansion problem, with different tradeoffs of the output set size.

Our results provide improved analyses of local graph partitioning algorithms when the vertex expansion or the $k$-way expansion is large,
and provide theoretical justification of their good empirical performances in applications such as image segmentation and data clustering (see~\cite{zhu} and the references therein).
The results show that the performances of local graph partitioning algorithms almost match that of Theorem~\ref{t:product} and Theorem~\ref{t:k-way} (within at most a $O(\log(|S|))$-factor in the approximation guarantee),
with the additional advantages that they apply to the small-set expansion problem (giving bicriteria approximations for $\phi_{\delta}(G)$) and also that their running time could be sublinear in the graph size (when $d$ and $|S|$ are small enough).

\subsection{Comparisons with Related Work}

\subsubsection{Generalizations of Cheeger's inequality}

There are several recent generalizations of Cheeger's inequality using higher eigenvalues of the Laplacian matrix.
The first generalization by Arora, Barak and Steurer~\cite{arora-barak-steurer} relates higher eigenvalues to small-set expansions:
\[\phi_{O(k^{-1/100})}(G) = O(\sqrt{\lambda_k \log_k n}),\]
and they use it to design a subexponential time algorithm for approximating unique games.
The second generalization by Louis et al.~\cite{louis+2} and Lee et al.~\cite{lee-gharan-trevisan} relates higher eigenvalues to $k$-way expansion (a stronger requirement than small-set expansion):
\begin{equation} \label{e:higher}
\frac12 \lambda_k \leq \phi_k(G) \leq O(\sqrt{\lambda_{2k} \log k}),
\end{equation}
and this justifies the use of higher eigenvalues in $k$-way graph partitioning.
Then there is a generalization by Kwok et al.~\cite{improved-cheeger} relating higher eigenvalues to the ordinary edge expansion:
\begin{equation} \label{e:improved}
\phi(G) \leq O(k) \frac{\lambda_2}{\sqrt{\lambda_k}},
\end{equation}
which shows that the spectral partitioning algorithm performs better in instances with $\lambda_k$ large for a small $k$.

Instead of using higher eigenvalues to give better bounds on expansion parameters, our results use expansion parameters to give better bounds on the second eigenvalue. 
We remark that the techniques developed in~\cite{improved-cheeger} could be used to prove Theorem~\ref{t:k-way} (see Section~\ref{s:appendix} in the Appendix),
but our approach is quite different and could be used to prove Theorem~\ref{t:product} and to extend Theorem~\ref{t:k-way} to analyze personal pagerank vectors. 
We also note that our proof of Theorem~\ref{t:k-way} can be used to prove (\ref{e:improved}) using a graph powering trick as described in~\cite{powers} (see Section~\ref{s:appendix} in the Appendix).

\subsubsection{Local Graph Partitioning Algorithms}

Local graph partitioning algorithms are useful in finding a small non-expanding set in a large graph, as their running times are only weakly dependent on the graph size and could be sublinear time.
All known algorithms are based on some random walks related processes.
The first local graph partitioning algorithm is a truncated random walk algorithm by Spielman and Teng~\cite{spielman-teng}, which returns a set $S'$ with $\phi(S') = O(\sqrt{\phi(S) \log^3 n})$ with work-to-volume ratio $O(\polylog(n)/\phi^2(S))$.
The second algorithm is a personal pagerank algorithm by Andersen, Chung and Lang~\cite{andersen-chung-lang}, which returns a set $S'$ with $\phi(S') = O(\sqrt{\phi(S) \log(|S|)})$ with work-to-volume ratio $O(\polylog(n)/\phi(S))$.
The evolving set process is used by Andersen and Peres~\cite{andersen-peres} to further improved the work-to-volume ratio to $O(\polylog(n)/\sqrt{\phi(S)})$ while having the same performance guarantee as in~\cite{andersen-chung-lang}.
Using a better analysis of the escaping probability of random walks, Oveis Gharan and Trevisan~\cite{gharan-trevisan} (see also~\cite{kwok-lau}) showed that the $\sqrt{\log(|S|)}$ factor in the performance ratio can be removed, thereby almost matching the guarantee of Cheeger's inequality.
They combined this with the evolving set process to find a set $S'$ with $\phi(S') = O(\sqrt{\phi(S)/\eps})$, $|S'| = O(|S|^{1+\eps})$ and work-to-volume ratio $O(
|S|^{\eps} \polylog(n) / \sqrt{\phi})$.

Our contribution is to show that the performance of some simple local graph partitioning algorithms (truncated random walks, personal pagerank) almost match that of the improved Cheeger's inequalities.
These provide the first analyses showing that random walk based algorithms perform better when $\ver(G)$, $\phi_k(G)$ or $\lambda_k(G)$ is large, 
with similar performances to the spectral partitioning algorithm while having additional features.
We note that Zhu et al.~\cite{zhu} gave a better analysis of the personal pagerank algorithm when the internal expansion of the target set is large; our results are related but incomparable.

\subsubsection{Analysis of Mixing Time}

The notion of expansion profile was introduced by Lov\'asz and Kannan~\cite{lovasz-kannan} in the study of mixing times of random walks.
They proved that the mixing time is upper bounded by $\int \frac{dx}{x \Phi(x)^2}$ where $\Phi(x) = \min_{0 \leq \delta \leq x} \phi_{\delta}$, which is a better bound on the mixing time when the average conductance is large (e.g. small sets expand well in geometric graphs).

Our work is inspired by their paper and some subsequent work~\cite{kannan-lovasz-montenegro, morris-peres}, both in the proof techniques (will be discussed in the next subsection) and in the definitions.
The robust vertex expansion and its expansion profile are studied in~\cite{kannan-lovasz-montenegro}, where better bounds on the mixing time are proved in a similar form to the average conductance bound above.
In particular, it implies the mixing time is bounded by $O(\log(n)/\Psi(G))$, and thus $\lambda_2 \geq \Omega(\Psi(G) / \log(n))$.
We note that Morris and Peres also proved a lower bound on the second eigenvalue (Theorem~15 in~\cite{morris-peres}) using a parameter related to vertex expansion, but their definition is incomparable to ours.

Our contribution is to directly bound the second eigenvalue (not the mixing time) using the expansion parameters and our bounds are independent of $n$.
Also, the bound that we prove using $\phi_k$ is considerably stronger.
Using $\phi_{\delta}$, the average conductance bound only gives $1/\lambda_2 \leq O(\log(\delta n)/\phi_{\delta}^2 + \log(1/\delta)/\phi^2)$, not improving on Cheeger's inequality even when $\phi_{\delta} = \Omega(1)$ for constant $\delta$, while Corollary~\ref{c:small-set} gives a $O(1)$-approximation when $\phi_{\delta} = \Omega(1)$ for constant $\delta$.

\subsection{Technical Overview}

The proofs are inspired by the work of Lov\'asz and Kannan~\cite{lovasz-kannan}.
We observe that their method is useful in analyzing the solution to a Laplacian system ($Lx=b$), and can be extended to study both the second eigenvectors ($Lx=\lambda x$) and the personal pagerank vectors.

The high-level approach is to look at the solution vector $x \in {\mathbb R}^n$ with $x_1 \geq x_2 \geq \ldots \geq x_n$, and relates the (slow) decrease of $x_i$ to the (large) expansion of the level sets in this vector.
Similar to~\cite{lovasz-kannan}, we define a jumping sequence of indices $1=m_0, m_1, m_2, \ldots$ such that $x_{m_{i}} - x_{m_{i+1}}$ is inversely proportional to the expansion of the level set $[1,m_i]$ (see Lemma~\ref{l:drop}).
Using the Laplacian equation of the second eigenvector,
we use an inductive argument to show that if the expansion of all level sets is $\Omega(\sqrt{\lambda_2})$, 
then the values of $x_i$ decrease slowly enough such that $x_{n/2} > 0$ (see Lemma~\ref{l:induction}), contradicting that $x$ is orthogonal to the all-one vector.
We remark that this gives a new and quite different proof of Cheeger's inequality (e.g.~without using the Cauchy-Schwarz inequality).
To prove Theorem~\ref{t:product}, we use the robust vertex expansion to argue that each jump can be made longer ($m_{i+1} - m_{i}$ made larger) and this gives the improved bound. 
To prove Theorem~\ref{t:k-way}, we argue that given an ordering of the vertices, if $\phi_k$ is large, then there are only a small number of indices in the jumping sequence whose corresponding level sets $[1,m_i]$ are of small expansion (see Lemma~\ref{l:jump}), and then we modify the induction hypothesis to obtain the result (see Lemma~\ref{l:induction2}).
The inductive arguments and the use of $\phi_k$ in arguing about expansions of level sets are the new elements in the proofs that improve upon the average conductance bound of Lov\'asz and Kannan.

The previous analyses of both the truncated random walk algorithm~\cite{spielman-teng} and the personal pagerank algorithm~\cite{andersen-chung-lang} are based on the combinatorial technique introduced by Lov\'asz and Simonovits~\cite{lovasz-simonovits} in analyzing the mixing time of random walks.
This technique is quite different from the analysis of spectral partitioning algorithms.
It requires to consider the random walk vectors for many different time steps, and it is difficult to incorporate the notions of $\phi_k$ or $\lambda_k$ in the analyses as the ordering and level sets are changing in each time step\footnote{We still don't know how to do a better analysis for the evolving set process because of this difficulty.}.

Our techniques provide two approaches to lift the analysis of the spectral graph partitioning algorithm for local graph partitioning algorithms, bringing closer the analyses of these two types of algorithms.
For the personal pagerank algorithm, we use the Lov\'asz-Kannan approach to directly analyze the vector so that we can use $\phi_k$ to reason about the level sets (Lemma~\ref{l:jump}).
We note that this approach is considered by Andersen and Chung to give a simplified proof of the personal pagerank algorithm~\cite{andersen-chung}, and we will reuse some of their lemmas to obtain Theorem~\ref{t:pagerank}.
For the truncated random walk algorithm, we use the spectral approach of Arora-Barak-Steurer~\cite{arora-barak-steurer} to directly obtain a vector with small Rayleigh quotient and small support, so that the improved Cheeger's inequalities can be applied to obtain results for approximating small-set expansions\footnote{We thank David Steurer for suggesting this spectral approach.}.

Finally, we remark that this approach can be applied to analyze the solutions to other Laplacian systems.
Consider the following algorithm for approximating edge expansion.
For an unknown target set $S$, pick a random vertex $s \in S$, inject $n$ units of current to $s$ and extracts one unit of current from every vertex in the graph, sort the vertices by the voltages\footnote{Or equivalently, sort the vertices based on the expected hitting time to $s$.}, and output the level set with the smallest expansion among all level sets of size up to $n/2$.
Our approach implies that this algorithm always outputs a set $S'$ with $\phi(S') = O(\sqrt{\phi(S) \log n})$.
We believe that this approach draws more connections to the mixing time literature and will find further applications.

\section{Spectral Partitioning}

Let $\lambda := \lambda_2$ and $x$ be a second eigenvector such that $Lx = \lambda x$.
Sort the vertices so that $x_1 \geq x_2 \geq \ldots \geq x_n$.

\subsection{Vertex Expansion}

The proof of Theorem~\ref{t:product} consists of two steps.
The first step is to prove the drop lemma and then define a jumping sequence to apply the lemma.
The second step is to use an inductive argument to derive a contradiction if the expansion of all level sets are large.

\subsubsection{Drop Lemma and Jumping Sequence}

The following lemma bounds the decrease of the values in $x$ to the expansion of the level sets of $x$.
Recall that $w(S,T)$ denotes the total weight of the edges with one vertex in $S$ and another vertex in $T$.

\begin{lemma}[Drop Lemma] \label{l:drop} 
For $1 \leq a < b \leq n$, we have
\[x_a - x_b \leq \frac{\lambda \sum_{i=1}^a x_i}{w([1,a],[b,n])}.\]
\end{lemma}
\begin{proof}
For each $i$, 
\[x_i - \sum_{j} w_{ij} x_j = \lambda x_i.\]
Sum this equation for $1 \leq i \leq a$, we have
\[\sum_{i=1}^a \sum_j x_i w_{ij} - \sum_{i=1}^a \sum_j x_j w_{ij} = \lambda \sum_{i=1}^a x_i.\]
Since $w_{ij} = w_{ji}$, this can be simplified to
\[\sum_{i \leq a} \sum_{j > a} x_i w_{ij} - \sum_{i \leq a} \sum_{j > a} x_j w_{ij} = \lambda \sum_{i=1}^a x_i.\]
Consider the edges from the set $[1,a]$ to the set $[b,n]$.
Each edge contributes $(x_i - x_j) w_{ij}$ to the left hand side,
which is at least $(x_a - x_b) w_{ij}$.
Therefore, we have
\[w([1,a],[b,n]) \cdot (x_a - x_b) \leq \lambda \sum_{i=1}^a x_i
\quad {\rm and~thus} \quad
x_a - x_b \leq \frac{\lambda \sum_{i=1}^a x_i}{w([1,a],[b,n])}.\]
\end{proof}

We define a jumping sequence of indices to apply the drop lemma.
Let $m_0=1$ and
\[m_{i+1} = \lceil m_i(1 + \ver(m_i)) \rceil,\]
where $\ver(m_i), \phi(m_i), \Psi(m_i)$ are shorthands for $\ver([1,m_i]), \phi([1,m_i]), \Psi([1,m_i])$ respectively.
Then, for $m_i \leq n/2$, by the definition of $\ver(m_i)$, we have
\[w([1,m_i],[m_{i+1},n])\geq \frac12 m_i \cdot \phi(m_i).\]
Putting it in the above inequality with $a=m_i$ and $b=m_{i+1}$, it follows that
\begin{equation} \label{e:jump}
x_{m_{i}} - x_{m_{i+1}} \leq \frac{2 \lambda \sum_{i=1}^{m_i} x_i}{m_i \cdot \phi(m_i)}
= \frac{2 \lambda \ol{x}_{m_i}}{\phi(m_i)}, \quad {\rm where} \quad 
\ol{x}_l := \frac{1}{l} \sum_{i=1}^l x_i.
\end{equation}
Note that $\ol{x}_l$ is non-increasing over $l$.

\subsubsection{Induction}

We will prove the following lemma by induction.

\begin{lemma} \label{l:induction}
If $\Psi(m_i) \geq 32\lambda$ and $\phi(m_i) \geq 32\lambda$ for all $m_i \leq n/2$,
then $\ol{x}_{m_{i+1}} \leq 2x_{m_{i+1}}$ for all $m_i \leq n/2$.
\end{lemma}

First we see how it implies Theorem~\ref{t:product}.
Let $m_j$ be the first term in the jumping sequence such that $m_j > n/2$.
Note that the assumptions of Lemma~\ref{l:induction} would imply that $x_j \geq \frac12 \ol{x}_j> 0$, where the last inequality follows because $\sum_{i=1}^n x_i = 0$ (as the second eigenvector is orthogonal to the all-one vector)
and so all partial sums are positive.
But this implies that $x_i > 0$ for all $1 \leq i \leq n/2$, and
applying the same argument to $-x$ will give us a contradiction.
Therefore, the assumptions of Lemma~\ref{l:induction} must not hold, and thus there is an $m_i \leq n/2$ with $\Psi(m_i) \leq 32\lambda$ or $\phi(m_i) \leq 32 \lambda$, proving Theorem~\ref{t:product}.

Now we proceed to prove Lemma~\ref{l:induction}.
It is clear that the inequality holds for $m_0$.
Assume that $\ol{x}_{m_i} \leq cx_{m_i}$ where $c=2$,\footnote{The variable $c$ is used so that we can reuse the calculation here for the proof of Theorem~\ref{t:k-way}.}
we would like to prove that $\ol{x}_{m_{i+1}} \leq cx_{m_{i+1}}$.
Note that
\[\sum_{i=1}^{m_{i+1}} x_i 
= \sum_{i=1}^{m_i} x_i + \sum_{i=m_{i}+1}^{m_{i+1}} x_i
\leq m_i \ol{x}_{m_i} + (m_{i+1} - m_i)x_{m_i} 
\leq x_{m_i}(m_{i+1} + (c-1)m_i). 
\]
Dividing both sides of this inequality by $m_{i+1}$, we have
\begin{eqnarray*}
\ol{x}_{m_{i+1}} & \leq & 
x_{m_i}(1 + (c-1)\frac{m_i}{m_{i+1}}) 
\leq x_{m_i}(1 + \frac{(c-1)}{1+\ver(m_i)}) \\
& \leq & x_{m_{i+1}} (\frac{c+\ver(m_i)}{1 + \ver(m_i)}) (\frac{\phi(m_i)}{\phi(m_i) - 2 \lambda c}) 
\leq c x_{m_{i+1}},
\end{eqnarray*}
where the second inequality follows from the definition of $m_{i+1}$,
the third inequality is by (\ref{e:jump}), and 
the last inequality follows from the following claim by plugging in $\ver(m_i)$ for $h$ and $\phi(m_i)$ for $\varphi$.
Note that the conditions of Claim~\ref{c:gao} follows from the assumptions of Lemma~\ref{l:induction}, and this completes the proof.

\begin{claim} \label{c:gao}
If $2 \leq c \leq 4$, $32\lambda \leq h\varphi$ and $32\lambda \leq \varphi$, then we have
\[(\frac{c+h}{1 + h}) (\frac{\varphi}{\varphi - 2 \lambda c}) \leq c.\]
\end{claim}
\begin{proof}
The conclusion to check is
\[c \geq (\frac{c+h}{1 + h}) (\frac{\varphi}{\varphi - 2 \lambda c})
= (c - \frac{(c-1)h}{1+h})(1 + \frac{2\lambda c}{\varphi - 2\lambda c}),
\]
which is equivalent to 
\[0 > \frac{-(c-1)h}{1 + h} + \frac{2\lambda c^2}{\varphi - 2\lambda c} - \frac{(c-1) h (2\lambda c)}{(1+h)(\varphi - 2 \lambda c)}.\]
Since $\varphi \geq 32 \lambda > 2c\lambda$, 
this is equivalent to
\[0 > -(c-1)h(\varphi - 2\lambda c) + 2\lambda c^2 (1 + h) - (c-1) h (2 \lambda c), \]
which can be simplified to
\[\frac{c-1}{c^2} > \frac{2 \lambda (1 + h)}{h\varphi} = 2\lambda(\frac{1}{h\varphi} + \frac{1}{\varphi}).\]
Since $2 \leq c \leq 4$, the left hand side is at least $1/8$.
We consider two cases.
The first case is when $1/(h\varphi) \geq 1/\varphi$,
and so the right hand side is at most $4\lambda/(h\varphi)$. 
We have $1/8 \geq 4\lambda/(h\varphi)$, as long as $h\varphi \geq 32 \lambda$,
which is satisfied by our assumption.
The second case is when $1/(h \varphi) \leq 1/\varphi$,
and so the right hand side is at most $4\lambda/\varphi$.
We have $1/8 \geq 4\lambda/\varphi$, as long as $\varphi \geq 32 \lambda$,
which is also satisfied by our assumption.
\end{proof}

\subsection{Proof of Theorem~\ref{t:k-way}}

We follow the same approach to prove Theorem~\ref{t:k-way}.
The additional arguments are in Lemma~\ref{l:jump} to bound the number of terms in the jumping sequence with small expansion using $\phi_k$ and in Claim~\ref{c:induction} to control the inductive bound dynamically.

For Theorem~\ref{t:k-way},
we define the jumping sequence as follows.
Let $m_0=1$ and
\[
m_{i+1} = \lceil m_i(1 + \frac12 \phi(m_i)) \rceil.
\]
Then, for $m_i \leq n/2$, we have
\[w([1,m_i],[m_{i+1},n])\geq m_i \phi(m_i) - (m_{i+1} - m_i - 1) \geq \frac12 m_i \cdot \phi(m_i),\]
so that equation (\ref{e:jump}) still holds after applying the drop lemma.

\subsubsection{$k$-way Expansion}

The assumption on $\phi_k$ allows us to bound the number of terms in the jumping sequence with small expansion.
We note that the following lemma can be applied to any ordering of vertices (not just for second eigenvector), and it will be applied to personal pagerank vectors later.

\begin{lemma} \label{l:jump}
For any $\theta < \phi_k/4$,
there are at most $16k/\phi_k$ terms $m_i$ in the jumping sequence with $\theta \leq \phi(m_i) \leq 2\theta$. 
\end{lemma}
\begin{proof}
Suppose by contradiction that there are at least $16k/\phi_k$ such terms.
Let $y_0$ be the first such term and let $y_{i}$ be the $(16i/\phi_k)$-th such term.
We claim that the sets
$\{[1,y_0], [y_0,y_1], \ldots, [y_{k-1},y_k]\}$ are all of expansion less than $\phi_k$, contradicting the definition of $\phi_k$.
Note that 
\[y_{i+1} \geq y_i (1 + \frac{\theta}{2})^{\frac{16}{\phi_k}} \geq y_i (1 + \frac{8 \theta}{\phi_k}),
\quad {\rm and~thus} \quad
y_{i+1} - y_i \geq \frac{8 \theta y_i}{\phi_k}.\]
The expansion of the set $[y_i,y_{i+1}]$ is
\begin{eqnarray*}
\phi([y_i,y_{i+1}]) & = & \frac{ w([y_i,y_{i+1}], [1,y_i] \cup [y_{i+1},n]) }{y_{i+1} - y_i} \\
& \leq & \frac{w([1,y_i],\overline{[1,y_i]}) + w([1,y_{i+1}],\overline{[1,y_{i+1}]})}{y_{i+1} - y_i}\\
& \leq & \frac{ 2\theta y_{i+1} + 2\theta y_i }{y_{i+1} - y_i}\\
& = & 2 \theta (1 + \frac{2y_{i}}{y_{i+1}-y_i}).
\end{eqnarray*}
Using the lower bound on $y_{i+1} - y_i$,
we have 
\[\phi([y_i,y_{i+1}]) \leq 2\theta + \frac{\phi_k}{2} < \phi_k,\]
where the last inequality is by our assumption that $\theta < \phi_k/4$. 
\end{proof}

\subsubsection{Induction}

In the following, we assume that $\phi_k^2 \geq 1024 \lambda$, 
as otherwise Theorem~\ref{t:k-way} holds trivially.
We will prove the following lemma by induction.

\begin{lemma} \label{l:induction2}
If $\phi_k^2 \geq 1024 \lambda$ and $\phi(m_i) \geq 256k\lambda/\phi_k$ for all $m_i \leq n/2$,
then $\ol{x}_{m_{i+1}} \leq 4x_{m_{i+1}}$ for all $m_i \leq n/2$.
\end{lemma}

As argued before, the assumptions of Lemma~\ref{l:induction2} would imply that $x_i > 0$ for all $1 \leq i \leq n/2$, leading to a contradiction.
So, the assumptions of Lemma~\ref{l:induction2} must not hold, and thus there is an $m_i$ with $\phi(m_i) \leq 256k\lambda/\phi_k$, proving Theorem~\ref{t:k-way}.

To prove Lemma~\ref{l:induction2},
we will prove by induction that $\ol{x}_{m_i} \leq c_i x_{m_i}$
where initially $c_0=2$ and 
\begin{equation*}
c_{i+1} = \left\{
	\begin{array}{rl}
		c_i & \text{if } \phi(m_i) \geq \phi_k / 4,\\
		c_i/(1- \eps_i c_i) & \text{if } \phi(m_i) < \phi_k / 4, \text{where } \eps_i = 2\lambda/\phi(m_i).
	\end{array} \right.
\end{equation*}
We first assume this induction step and show that $c_{\infty} \leq 4$ using Lemma~\ref{l:jump}.
Then we will verify the induction step.

\begin{claim} \label{c:induction}
$c_{\infty} \leq 4.$
\end{claim}
\begin{proof}
First, we prove by induction that 
\[c_{i} = \frac{c_0}{1 - \sum_{j=0}^{i-1} \eps_j c_0}.\]
Assume this is true for $i$.
Then 
\[c_{i+1} = \frac{c_i}{1-\eps_i c_i} = 
(\frac{c_0}{1 - \sum_{j=0}^{i-1} \eps_j c_0})
(\frac{1}{1- \eps_i (\frac{c_0}{1 - \sum_{j=0}^{i-1} \eps_j c_0})})
= \frac{c_0}{1 - \sum_{j=0}^{i} \eps_j c_0}.\]
Next, we bound $c_{\infty}$ using Lemma~\ref{l:jump}.
Recall that $\eps_i = 2\lambda / \phi(m_i)$
and we can assume that $\phi(m_i) \geq 256k\lambda/\phi_k$.
Let $\theta_0 = 256k\lambda/\phi_k$ and $\theta_{i+1}=2\theta_i$.
By Lemma~\ref{l:jump},
there are at most $16k/\phi_k$ terms $m_i$ in the jumping sequence with $\theta \leq \phi(m_i) < 2 \theta$ when $\theta < \phi_k/4$.
Therefore,
\begin{eqnarray*}
\sum \eps_j & = & 
\sum_{i \geq 0} \sum_{j: \theta_i \leq \phi(m_j) \leq 2 \theta_i} \eps_j 
\leq \sum_{i \geq 0} \sum_{j: \theta_i \leq \phi(m_j) \leq 2 \theta_i} \frac{2\lambda}{\theta_i} 
 \leq  \sum_{i \geq 0} \frac{16k}{\phi_k} \frac{2\lambda}{\theta_i}
= \sum_{i \geq 0} \frac{32k\lambda}{\phi_k} \frac{\phi_k}{256k\lambda 2^i}
= \frac{1}{4}.
\end{eqnarray*}
Therefore,
\[c_{\infty} = \frac{c_0}{1 - \sum_j \eps_j c_0} \leq \frac{c_0}{1-\frac{c_0}{4}} = 4.\]
\end{proof}

We prove the induction step.
There are two cases, depending on whether $\phi(m_i) < \phi_k/4$.
We first consider the case when $\phi(m_i) < \phi_k/4$.
In this case, just apply equation (\ref{e:jump}) and we have
\[x_{m_{i+1}} \geq x_{m_i} - \frac{2\lambda c_i x_{m_i}}{\phi(m_i)}
= x_{m_i}(1 - \eps_i c_i) 
\geq \ol{x}_{m_i} (\frac{1-\eps_i c_i}{c_i})
\geq \frac{\ol{x}_{m_{i+1}}}{c_{i+1}},\]
by the definition of $\eps_i$ and $c_{i+1}$
and we are done in this case.

It remains to consider the case when $\phi(m_i) \geq \phi_k / 4$.
By induction, we assume that $\ol{x}_{m_i} \leq c_i x_{m_i}$, 
and we claim that $\ol{x}_{m_{i+1}} \leq c_i x_{m_{i+1}}$.
By the same calculation as in the induction for Theorem~\ref{t:product}, 
we have 
\[\sum_{i=1}^{m_{i+1}} x_i 
= \sum_{i=1}^{m_i} x_i + \sum_{i=m_{i}+1}^{m_{i+1}} x_i
\leq x_{m_i}(m_{i+1} + (c_i-1)m_i).\]
Similarly, dividing both sides of this inequality by $m_{i+1}$, we have
\[
\ol{x}_{m_{i+1}}  
\leq x_{m_{i+1}} (\frac{c_i + \frac12 \phi(m_i)}{1 + \frac12 \phi(m_i)}) (\frac{\phi(m_i)}{\phi(m_i) - 2 \lambda c_i}) \leq c_i x_{m_{i+1}},
\]
where the last inequality follows from Claim~\ref{c:gao} by plugging in $h=\phi(m_i)/2$, $\varphi=\phi(m_i)$, $c=c_i$ and checking that the conditions $2 \leq c \leq 4$ (Claim~\ref{c:induction}), $h\varphi \geq \varphi^2/2 \geq \phi_k^2 / 32 \geq 32\lambda$ and $\varphi \geq 32 \lambda$ are satisfied by our assumptions. 
This completes the induction step and thus the proof of Lemma~\ref{l:induction2}.

\section{Personal Pagerank}

We show that a similar and simpler analysis applies to the personal pagerank vector.
Given a parameter $\alpha \in (0,1]$ and a vertex $s$,
the personal pagerank vector $r_{s,\alpha} \in {\mathbb R}^n$ is the unique solution to the equation $r_{s,\alpha} = \alpha \chi_s + (1 - \alpha)Wr_{s,\alpha}$, where $W$ is the transition matrix of the lazy random walks.
Note that $r_{s,\alpha}$ is a probability distribution vector.
In the following, 
we assume $S$ is an unknown target set with $3|S| \log(|S|) \leq n$.

\subsection{Drop Lemma}

Let $x := r_{s,\alpha}$ be the personal pagerank vector and assume $x_1 \geq x_2 \geq \ldots \geq x_n$.
Andersen and Chung proved a drop lemma for pagerank vectors (see Lemma~1 of~\cite{andersen-chung} and compared to our Lemma~\ref{l:drop}), 
for $1 \leq a < b \leq n$,
\begin{equation} \label{e:pagerank-drop}
x_a - x_b \leq \frac{\alpha}{w([1,a],[b,n])}.
\end{equation}

\subsection{Escaping Probability}

Let $S$ be an unknown target set.
Using a bound on the escaping probability of random walks~\cite{spielman-teng}\footnote{Actually, using a stronger result by Oveis Gharan and Trevisan~\cite{gharan-trevisan}, one can show that
$x(S) \geq \frac{\phi(S)(1 + \alpha)}{\alpha + \phi(S)(1 - \alpha)}$,
but it does not change the results in the following subsections.},
Andersen and Chung proved that for half of the vertices $s$ in $S$,
the personal pagerank vector $x := r_{s,\alpha}$ will have the property that
(see Lemma~5 of~\cite{andersen-chung})
\begin{equation} \label{e:pagerank-escape}
\sum_{i \in S} x_i \geq 1 - \frac{\phi(S)}{\alpha}.
\end{equation}
Setting $\alpha = 3\phi(S)$ makes sure that $\sum_{i \in S} x_i \geq 2/3$ and it follows that (see Lemma~2 of~\cite{andersen-chung}) there exists an $a \leq |S|$ with 
\[x_a \geq \frac{2}{3 a \log(|S|)}.\]

\subsection{Vertex Expansion}

For vertex expansion, we start our jumping sequence by setting $m_0=a$ and then define 
\[m_{i+1} = \lceil m_i (1 + \ver(G)) \rceil.\]
By this definition, we have $w([1,m_i],[m_{i+1},n]) \geq \frac12 m_i \cdot \phi(m_i)$, and it follows that
\[ x_{m_{i+1}} \geq x_{m_i} - \frac{2\alpha}{m_i \cdot \phi(m_i)} 
\quad {\rm and} \quad
x_{m_{\infty}}  \geq  x_a - \sum_{i \geq 0} \frac{2\alpha}{m_i \cdot \phi(m_i)}.\]
Suppose by contradiction that $\phi(m_i) \geq 36\phi(S)\log(|S|)/\ver(G)$ for all $m_i \leq 3 |S|\log(|S|)$.
Then 
\[\sum_{i \geq 0} \frac{2\alpha}{m_i \cdot \phi(m_i)} 
\leq \sum_{i \geq 0} \frac{2\alpha \ver(G)}{36a(1 + \ver(G))^i \phi(S) \log(|S|)} 
\leq \frac{1}{3a \log(|S|)},\]
where the last inequality uses the bound that 
$\sum_{i \geq 0} 1/(1+\ver(G))^i \leq (1+\ver(G))/\ver(G) \leq 2/\ver(G)$ 
and our choice that $\alpha = 3\phi(S)$.
This implies that
\[x_{3|S|\log(|S|)} \geq \frac{1}{3a \log(|S|)} 
\quad {\rm and~thus} \quad
\sum_{j \geq 0} x_j \geq 
\sum_{0 \leq j \leq 3|S|\log(|S|)} \frac{1}{3a\log(|S|)}> 1,
\]
since $a \leq |S|$, contradicting that $x$ is a probability distribution vector.
Therefore, there must exist an $m_i \leq 3|S| \log (|S|)$ with $\phi(m_i) \leq 36\phi(S) \log(|S|)/ \ver(G)$, proving the first part of Theorem~\ref{t:pagerank}.

\subsection{$k$-way Expansion}

For $k$-way expansion, we define the jumping sequence by setting $m_0=a$ and 
\[m_{i+1} = \lceil m_i (1 + \phi(m_i)) \rceil.\]
As before, we have $w([1,m_i],[m_{i+1},n]) \geq \frac12 m_i \cdot \phi(m_i)$, and it follows that
\[ x_{m_{i+1}} \geq x_{m_i} - \frac{2\alpha}{m_i \cdot \phi(m_i)} 
\quad {\rm and} \quad
x_{m_{\infty}}  \geq  x_a - \sum_{i \geq 0} \frac{2\alpha}{m_i \cdot \phi(m_i)}.\]
We divide the summation into two parts
\[ \sum_{i: \phi(m_i) < \phi_k/4} \frac{2\alpha}{m_i \cdot \phi(m_i)}
+ \sum_{i: \phi(m_i) \geq \phi_k/4} \frac{2\alpha}{m_i \cdot \phi(m_i)}.\]
The second part is at most
\[\sum \frac{2\alpha}{m_i \phi_k/4} 
\leq \sum \frac{8\alpha}{a \phi_k (1 + \phi_k/4)^i} 
\leq \frac{64 \alpha}{a \phi_k^2}.\]
The first part can be bounded by Lemma~\ref{l:jump} as follows.
Suppose by contradiction that $\phi(m_i) \geq 1152k \phi(S) \log(|S|) / \phi_k$ for all $m_i \leq 3 |S|\log(|S|)$.
Let $\theta_0 = 1152k \phi(S) \log(|S|) / \phi_k$ and $\theta_{i} = 2\theta_{i-1}$ for $i \geq 1$.
By Lemma~\ref{l:jump}, there are at most $16k/\phi_k$ terms in the jumping sequence having conductance between $\theta$ and $2\theta$ when $\theta < \phi_k/4$.
Therefore, the first part is at most
\[\sum_{j} \sum_{i: \theta_j \leq \phi(m_i) \leq 2\theta_j} \frac{2\alpha}{a \phi(m_i)}
\leq \sum_j \frac{32k}{\phi_k} \frac{\alpha}{a \theta_j} 
= \sum_j \frac{32k}{\phi_k} \frac{\alpha}{a 2^j \theta_0}
= \frac{64k\alpha}{a \phi_k \theta_0}.
\]
Putting these back into the first inequality, we have
\begin{eqnarray*}
x_{m_{\infty}} & \geq & x_a - \frac{64 \alpha}{a \phi_k^2} - \frac{64k\alpha}{a \phi_k \theta_0}
\geq  \frac{2}{3 a \log(|S|)} - \frac{192 \phi(S)}{a \phi_k^2} - \frac{192 \phi(S) k}{a \phi_k \theta_0}\\
& \geq & \frac{2}{3 a \log(|S|)} - \frac{384 \phi(S) k}{a \phi_k \theta_0}
 \geq  \frac{1}{3a \log(|S|)},
\end{eqnarray*}
where the second inequality is by the lower bound of $x_a$ and the choice of $\alpha=3\phi(S)$,
and the last inequality is by our choice of $\theta_0$.
This implies that
\[x_{3|S|\log(|S|)} \geq \frac{1}{3a \log(|S|)} 
\quad {\rm and~thus} \quad
\sum_{j \geq 0} x_j \geq 
\sum_{0 \leq j \leq 3|S|\log(|S|)} \frac{1}{3a\log(|S|)}> 1,
\]
since $a \leq |S|$, contradicting that $x$ is a probability distribution vector.
Therefore, there must exist an $m_i \leq 3|S| \log (|S|)$ with $\phi(m_i) \leq 1152k\phi(S) \log(|S|)/ \phi_k$, proving the second part of Theorem~\ref{t:pagerank}.

\subsection{Local Algorithm}

Andersen and Chung~\cite{andersen-chung} show that the drop lemma (equation (\ref{e:pagerank-drop})) still holds even for approximate personal pagerank vectors, which can be computed efficiently in unweighted graphs.
In the following, we assume the graphs are unweighted $d$-regular (in our setting, the edge weights are either $1 / d$ or $0$).
An $\epsilon$-approximate vector for $r_{s, \alpha}$ is a vector $r'_{s, \alpha}$ that satisfies $r'_{s, \alpha} = \alpha (\chi_s - q) + (1 - \alpha) W r'_{s, \alpha}$ where the vector $q$ is non-negative and satisfies $q(u) \le \epsilon$ for every vertex $u$ in the graph.

\begin{lemma}[\cite{andersen-chung}]
There is an algorithm that computes an $\epsilon$-approximate vector $r'_{s, \alpha}$.
The running time of the algorithm is $O(d / (\epsilon \alpha))$.
Assume $r'_{s,\alpha}(1) \ge r'_{s,\alpha}(2) \ge \dots \ge r'_{s,\alpha}(n)$.
The approximate vector $r'_{s,\alpha}$ satisfies for any $1 \le a < b \le n$,
\[
r'_{s,\alpha}(a) - r'_{s,\alpha}(b)
\le \frac\alpha{w([1, a], [b, n])}.
\]
\end{lemma}

Note that $0 \le q \le \epsilon \vec 1$ implies
\[
r_{s, \alpha} - r'_{s, \alpha}
= \alpha (I - (1 - \alpha) W)^{-1} (\chi_s - (\chi_s - q))
= \alpha (I - (1 - \alpha) W)^{-1} q
\le \epsilon \alpha (I - (1 - \alpha) W)^{-1} \vec1
= \epsilon \vec1,
\]
where the last equality holds since $\vec1$ is an eigenvector of both $I$ and $W$ with eigenvalue $1$.
Hence for any vertex $u$, we have $r'_{s, \alpha}(u) \ge r_{s, \alpha}(u) - \epsilon$.
We set $\epsilon = 1 / (6 |S|)$ and $\alpha = 3 \phi(S)$ so that
\[
\sum_{i \in S} r'_{s,\alpha}(i)
\ge \sum_{i \in S} r_{s, \alpha}(i) - \epsilon |S|
\ge 1 - \frac{\phi(S)}\alpha - \frac16
\ge \frac12,
\]
for those vertices $s$ that satisfy equation~(\ref{e:pagerank-escape}).
Hence there exists an $a \le |S|$ with
\[
r'_{s,\alpha}(a) \ge \frac1{2a \log |S|}.
\]

Since $r'_{s,\alpha}$ satisfies the drop lemma (equation (\ref{e:pagerank-drop})) and contains good initial value, both arguments in vertex expansion and k-way expansion follow (with the assumption $3|S|\log(|S|) \leq n$ replaced by $6|S|\log(|S|) \leq n$).
The runtime of this algorithm is dominated by the runtime for computing the approximate vector $p'_{s, t}$ and sorting at most $O(|S| \log(|S|))$ vertices after, and hence the total complexity is
$O(d / (\epsilon \alpha) + |S| \log^2(|S|)) = O(d |S| / \phi(S) + |S| \log^2(|S|))$.

\section{Random Walks} \label{s:walks}

In this section, we present a spectral analysis of the random walk local graph partitioning algorithm~\cite{spielman-teng, kwok-lau}.
The proof consists of three steps.
The first step is to show that the Rayleigh quotient of the random walk vector $p_{s,t} = W^t \chi_s$ is small, by using the analysis in the power method.
The second step is to show that the $\|p_{s,t}\|_2$ is large for many vertices $s$ in the unknown target set, by using the bound on escaping probability (or the staying probability).
This allows us to apply the argument in~\cite{arora-barak-steurer} to $p_{s,t}$ to obtain a vector with small Rayleigh quotient and small support.
Then we apply the improved Cheeger's inequality to prove Theorem~\ref{t:walks}.
Finally, we show that the truncated random walk vectors would also work, thereby proving a local implementation of the algorithm.

\subsection{Rayleigh Quotient}

Recall that the Rayleigh quotient of a vector $x$ is defined as $\rquo(x) = x^T L x / \|x\|^2$.
The following lemma shows that the Rayleigh quotient of the vector $p_{s,t} := W^t \chi_s$ becomes smaller when $t$ becomes larger.
The proof follows the analysis of the power method in computing the largest eigenvector.
\begin{lemma}
\label{l:small_rayleigh_quotient}
For any starting vertex $s$,
\[
\rquo(p_{s,t}) \le 2 - 2 \| p_{s,t} \|_2^{1 / t}.
\]
\end{lemma}
\begin{proof}
Let $\chi_s = \sum_{i = 1}^{n} c_i v_i$ where $v_i$ are eigenvectors of $L$.
Note that the lazy random walk matrix is $W = I - L/2$,
and thus the vector 
$p_{s,t} = W^t \chi_s = \sum_{i=1}^n c_i (1 - \lambda_i/2)^t v_i$.
Hence, the Rayleigh quotient of $p_{s,t}$ is
\[
\rquo(p_{s,t})
= \frac{p_{s,t}^T L p_{s,t}}{\|p_{s,t}\|^2}
= \frac{\sum_{i = 1}^{n} c_i^2 (1 - \lambda_i / 2)^{2t} \lambda_i}{\sum_{i = 1}^{n} c_i^2 (1 - \lambda_i / 2)^{2t}}
= 2 - 2 \frac{\sum_{i = 1}^{n} c_i^2 (1 - \lambda_i / 2)^{2t + 1}}{\sum_{i = 1}^{n} c_i^2 (1 - \lambda_i / 2)^{2t}}.
\]
Note that $\sum_i c_i^2 = \| \chi_s \|_2^2 = 1$, and thus $c_i^2$ can be viewed as a probability distribution.
Let $X$ be the random variable having value $1 - \lambda_i / 2$ with probability $c_i^2$.
Then we can write $\rquo(p_{s,t}) = 2 - 2 \expe[X^{2t + 1}] / \expe[X^{2t}]$.
By the power mean inequality and the non-negativity of $X$, we have
\[
\expe[X^{2t + 1}]^{1 / (2t + 1)} \ge \expe[X^{2t}]^{1 / (2t)}.
\]
Hence
\[
\rquo(p_{s,t})
\le 2 - 2 \expe[X^{2t}]^{1 / 2t}
= 2 - 2 \left( \sum_{i = 1}^{n} c_i^2 (1 - \lambda_i / 2)^{2t} \right)^{1 / 2t}
= 2 - 2 \| p_{s,t} \|_2^{1 / t}.
\]
\end{proof}

\subsection{Small Support Vector with Small Rayleigh Quotient}

A vector $x$ is called spectrally $\delta$-sparse if $\| x \|_1^2 \le \delta n \| x \|_2^2$.
First, by using a result by Oveis Gharan and Trevisan on escaping probability (or staying probability), we bound the spectral sparsity of the random walk vector.
Then, we use a result used by Arora, Barak and Steurer to turn a spectrally sparse vector into a small support vector with similar Rayleigh quotient.

The following lemma by Oveis Gharan and Trevisan shows that if $\phi(S)$ is small, there is a large subset $U \subseteq S$, such that the random walk starting at any vertex $s \in U$ stays entirely inside $S$ with good probability.
In particular, the probability that the walk ends inside $S$ is large.
\begin{theorem}[\cite{gharan-trevisan}]
\label{t:staying_probability}
For any subset $S \subseteq V$, there is a subset $U \subseteq S$, such that $|U| \ge |S| / 2$, and for any $s \in U$ we have
\[
\sum_{v \in S} p_{s, t}(v) \ge \frac1{200} \left( 1 - \frac{3 \phi(S)}2 \right)^t.
\]
\end{theorem}
This provides a bound on the spectral sparsity of $p_{s,t}$.
\begin{lemma}
\label{l:spectrally_sparse}
For any subset $S \subseteq V$, there is a subset $U \subseteq S$ such that $|U| \geq |S|/2$, and for any $s \in U$ we have
\[
\| p_{s,t} \|_1^2
\le \frac{40000 |S|}{(1 - 3 \phi(S) / 2)^{2t}} \| p_{s,t} \|_2^2.
\]
\end{lemma}
\begin{proof}
By Cauchy-Schwarz and Theorem~\ref{t:staying_probability}, we have
\[
\| p_{s,t} \|_2^2
\geq \sum_{v \in S} p_{s,t}(v)^2
\geq \frac{1}{|S|} \left( \sum_{v \in S} p_{s,t}(v) \right)^2
\ge \frac1{|S|} \left( \frac1{200} \left( 1 - \frac{3 \phi(S)}2 \right)^t \right)^2
= \frac1{40000 |S|} \left( 1 - \frac{3 \phi(S)}2 \right)^{2t}.
\]
Since $\|\chi_s\|_1 = 1$ and $W$ preserves $1$-norm, we have $\| p_{s,t} \|_1^2 = \|W^t \chi_s\|^2 = 1$, and the result follows.
\end{proof}

The following lemma in~\cite{arora-barak-steurer} shows how to obtain a vector $y$ with small support and similar Rayleigh quotient from a spectrally $\delta$-sparse vector $x$.
The proof is by choosing an appropriate threshold $t$ and set $y = \max(x-t,0)$.
\begin{lemma}[\cite{arora-barak-steurer}]
\label{l:spectral_rounding}
Let $x \in \mathbb R_{\ge 0}^{|V|}$ be a non-negative vector with $\| x \|_1^2 \le \delta n \| x \|_2^2$.
Then there exists a vector $y$ with $\supp(y) = O(\delta n)$
and $\rquo(y) = O(\rquo(x))$.
\end{lemma}

We will apply Lemma~\ref{l:spectral_rounding} on $p_{s,t}$ and obtain a vector with small Rayleigh quotient (Lemma~\ref{l:small_rayleigh_quotient}) and small support (Lemma~\ref{l:spectrally_sparse}).

\subsection{Improved Cheeger's Guarantees} \label{s:walks-improved}

We are ready to prove Theorem~\ref{t:walks}.
In the following we assume $\phi(S) \le 1 / 4$ and $|S| \ge 2$.
We set $t = \epsilon \log |S| / (6 \phi(S))$ so that
\[
\left( 1 - \frac{3 \phi(S)}2 \right)^{2t}
\ge \exp(-3 \phi(S))^{2t}
= \exp(-6 t \phi(S))
= \exp(-\epsilon \log |S|)
= |S|^{-\epsilon}.
\]
By Lemma~\ref{l:spectrally_sparse}, we have
\[
\| p_{s,t} \|_1^2 \le 40000 |S|^{1 + \epsilon} \| p_{s,t} \|_2^2.
\]
On the other hand, since $|S| \ge 2$,
\[
(40000 |S|)^{-1 / (2t)}
\ge \exp \left( -\frac{17 \log |S|}{2t} \right)
= \exp \left( -\frac{51 \phi(S)}\epsilon \right)
\ge 1 - \frac{51 \phi(S)}\epsilon.
\]
Therefore, by Lemma~\ref{l:small_rayleigh_quotient}, we have
\[
\rquo(p_{s,t})
\le 2 - 2 (1 - \frac{3 \phi(S)}2)(1 - \frac{ 51 \phi(S)}{2 \epsilon})
= O \left( \frac{\phi(S)}{\epsilon} \right).
\]
Now, we apply Lemma~\ref{l:spectral_rounding} by plugging the vector $p_{s,t}$ for $x$ and obtain a vector $y$ with $\supp(y) \le O(|S|^{1 + \epsilon})$ and $\rquo(y) = O(\phi(S) / \epsilon)$.
Finally, by the proof of the improved Cheeger's inequality~(\ref{e:improved}) (see Section~\ref{s:appendix} in Appendix), we find a level set $S'$ with $|S'| \le |\supp(y)| = O(|S|^{1 + \epsilon})$ and
\[
\phi(S') = O \left( \frac{k \phi(S)}{\epsilon \phi_k} \right)
\quad {\rm or} \quad
\phi(S') = O \left( \frac{k \phi(S)}{\epsilon \sqrt{\lambda_k}} \right).
\]
Since a level set of $y$ is a level set of $p_{s,t}$, 
this proves the approximation guarantee of Theorem~\ref{t:walks}.

\subsection{Local Algorithm}

Computing the vector $p_{s,t} = W^t \chi_s$ exactly requires at least linear time.
In the following, we assume the graph is an unweighted $d$-regular graph (in our setting, the edge weight is either $1/d$ or $0$).
To obtain a local algorithm,
we can compute a good approximation to $p_{s,t}$ by repeatedly applying the operator $W$ (initially we compute $W \chi_s$) and truncating the small values to zero.
\begin{lemma}[\cite{spielman-teng, kwok-lau}]
\label{l:compute_approximate_vector}
Let $p_{s,t} = W^t \chi_s$ be the exact random walk vector starting at vertex $s$.
There is an algorithm that compute a vector $p'_{s,t}$ such that $p_{s,t} \ge p'_{s,t} \ge p_{s,t} - \alpha \vec{1}$ and $p'_{s,t} \ge 0$ in time $O(d t^2 / \alpha)$.
\end{lemma}

We set $t = \epsilon \log |S| / \phi(S)$ and $\alpha = \phi(S) / (160000 |S|^{1 + \epsilon})$, so that the time complexity of our local algorithm is $O(d \epsilon^2 |S|^{1 + \epsilon} \log^2 |S| / \phi(S)^3)$.
It remains to show that $p'_{s,t}$ is still spectrally sparse and has small Rayleigh quotient.

\begin{lemma}
\label{l:approximate_spectrally_sparse}
For $p_{s,t}$ that satisfies the conclusion in Lemma~\ref{l:spectrally_sparse}, we have
\[
\| p'_{s,t} \|_1^2 \le \frac1{80000 |S|^{1 + \epsilon}} \| p'_{s,t} \|_2^2.
\]
\end{lemma}
\begin{proof}
In the proof, we let $x := p_{s,t}$ and $y := p'_{s,t}$.
By Lemma~\ref{l:compute_approximate_vector}, we have $y(i)^2 \ge x(i)^2 - 2 \alpha x(i)$ since $y(i) \ge \max(x(i) - \alpha, 0)$.
Therefore,
\[
\| y \|_2^2
= \sum_i y(i)^2
\ge \sum_i x(i)^2 - 2 \alpha \sum_i x(i)
= \| x \|_2^2 - 2 \alpha.
\]
By Lemma~\ref{l:spectrally_sparse}, we have $\| x \|_2^2 \ge 1 / (40000 |S|^{1 + \epsilon})$.
From our choice of $\alpha$, we have $2 \alpha = \phi(S) / (80000 |S|^{1 + \epsilon}) \le \phi(S) \| x \|_2^2 / 2$.
Therefore,
\[
\| x \|_2^2 - 2 \alpha
\ge \| x \|_2^2 \left( 1 - \frac{\phi(S)}2 \right)
\ge \frac1{80000 |S|^{1 + \epsilon}}
\ge \frac1{80000 |S|^{1 + \epsilon}} \| y \|_1^2,
\]
where the last inequality holds as $\| y \|_1^2 \le \| x \|_1^2 = 1$.
\end{proof}

\begin{lemma} \label{l:truncated-rayleigh}
\[
\rquo(p'_{s,t})
\le O(\frac{\phi(S)}{\epsilon}).
\]
\end{lemma}
\begin{proof}
Again, we let $x := p_{s,t}$ and $y := p'_{s,t}$ in the proof.
Let $r = x - y \ge 0$.
Then we have
\[
\rquo(y)
= \frac{y^T L y}{y^T y}
= \frac{(x - r)^T L (x - r)}{y^T y}
= \frac{x^T L x + r^T L r - 2 x^T L r}{y^T y}
\le \frac{2 x^T L x + 2 r^T L r}{y^T y}.
\]
By the calculation in Lemma~\ref{l:approximate_spectrally_sparse}, $\| y \|_2^2 \ge (1 - \phi(S) / 2) \| x \|_2^2$.  
Hence, using $r \ge 0$ and $y \ge 0$, we have $\| r \|_2^2 \le \| x \|_2^2 - \| y \|_2^2 \le \phi(S) \| x \|_2^2 / 2$ and $r^T L r \le 2 r^T r \le \phi(S) \| x \|_2^2$.
So, we have
\[
\rquo(y)
= O \left( \frac{x^T L x}{y^T y} + \frac{r^T L r}{y^T y} \right)
= O \left( \frac{x^T L x}{x^T x} + \frac{r^T L r}{x^T x} \right)
= O (\rquo(x) + \phi(S))
= O \left( \frac{\phi(S)}\epsilon \right).
\]
\end{proof}

With Lemma~\ref{l:approximate_spectrally_sparse} and Lemma~\ref{l:truncated-rayleigh}, we can use the same proof in Section~\ref{s:walks-improved} to prove Theorem~\ref{t:walks} with the time complexity claimed.

To prove Theorem~\ref{t:walks}(2), we only need to set $\epsilon = 1 / \log(|S|)$ so that $|S|^{1 + \epsilon} = O(|S|)$.

\subsection{Local Eigenvalue}

We remark that if we do not care about local implementations, we can find a particular good starting vertex $u$ such that the random walk algorithm starting at $u$ gives a better performance guarantee $\phi(S') = O(k \lambda_S / (\epsilon \phi_k))$,
where $\lambda_S$ is the smallest eigenvalue of the matrix $L_S$ which is the restriction of $L$ on the subset $S$.
Chung~\cite{Chung07} shows the following local Cheeger's inequality:
\[
\lambda_S \le \min_{T \subseteq S} \phi(T) \le \sqrt{2 \lambda_S}.
\]
Hence $\lambda_S$ is at most $\phi(S)$ and could be much smaller,
for instance when a subset of $S$ has very small expansion.
The idea is similar to that in~\cite{kwok-lau} and we just give a quick sketch.
Let $v_S$ be the corresponding eigenvector with eigenvalue $\lambda_S$.
We choose our starting vertex to be $u = \argmax_i |v_S(i)|$.
Then we show that the spectral sparsity of the $t$-steps random walk is at most $|S| / (1 - \lambda_S)^{2t} < |S| / (1 - O(\phi(S)))^{2t}$.
This allows us to set $t$ to be larger so as to improve the Rayleigh quotient of the random walk vector.

\section*{Acknowledgement}

Part of the work was done while we were long-term participants in the Algorithmic Spectral Graph Theory program at the Simons Institute for the Theory of Computing in Fall 2014.
We thank the organizers for the support and the wonderful research environment.
We also thank David Steurer for suggesting the spectral approach to analyze random walks, and Luca Trevisan for pointing out that Theorem~\ref{t:k-way} can be derived from the proof of the improved Cheeger's inequality in~\cite{improved-cheeger}.

\bibliographystyle{plain}

\appendix

\section{Relations between Improved Cheeger's Inequality and Theorem~\ref{t:k-way}} \label{s:appendix}

First, we show that Theorem~\ref{t:k-way} can be derived from the proof of the improved Cheeger's inequality, as pointed out to us by Luca Trevisan.
Then, we show that the improved Cheeger's inequality can be derived from Theorem~\ref{t:k-way}, using a graph powering trick as described in~\cite{powers}.

\subsection{Improved Cheeger's inequality implies Theorem~\ref{t:k-way}}

The following stronger statement was shown in~\cite{improved-cheeger}.

\begin{theorem}[Theorem 3.5 of \cite{improved-cheeger}, restated] \label{t:improved}
For any non-negative vector $x$ with $\supp(x) \le n / 2$, let $\phisw(x)$ be the minimum expansion of the level sets of $x$.
At least one of the following holds:
\begin{enumerate}
\item $\phisw(x) \le O(k) \rquo(x)$.
\item There exists $k$ disjointly supported vectors $x_1, \dots, x_k$ such that for all $1 \le i \le k$, $\supp(x_i) \subseteq \supp(x)$ and $\rquo(x_i) \le O(k^2 \rquo(x)^2 / \phisw(x)^2)$.
\end{enumerate}
\end{theorem}

We apply the theorem with $x = \max(v_2, 0)$ or $x = \max(-v_2, 0)$, whichever of smaller support.
Note that $\rquo(x) \le \lambda_2$ by standard argument~\cite{hoory-linial-wigderson}.
When the first case of Theorem~\ref{t:improved} holds, it is clear that
\[
\phi(G)
\le \phisw(x)
\le O(k) \rquo(x)
\le O(k \lambda_2)
\le O(\frac{k \lambda_2}{\phi_k}).
\]
Otherwise, there exist $k$ disjointly supported vectors, each with Rayleigh quotient not larger than $O(k^2 \lambda_2^2 / \phisw^2(x))$.
Apply (the original) Cheeger's arguments on these vectors, we can find $k$ disjoint subsets $S_i$, each of them satisfies $\phi(S_i) \le O(k \lambda_2 / \phisw(x))$.
This implies that
\[
\phi_k
\le O(\frac{k \lambda_2}{\phisw(x)}),
\quad {\rm or } \quad
\lambda_2
= \Omega(\frac{\phi_k \phisw(x)}k)
= \Omega(\frac{\phi_k \phi(G)}k),
\]
and Theorem~\ref{t:k-way} follows.

\subsection{Theorem~\ref{t:k-way} implies improved Cheeger's inequality}

In \cite{powers}, the authors proved a lower bound on the expansion of graph powers and used it to show some reductions on Cheeger's inequalities.
We show that the same approach can be used to prove improved Cheeger's inequality by Theorem~\ref{t:k-way}.

\begin{theorem}[Theorem 1 of \cite{powers}, restated]
\label{t:sqrtt}
Let $H$ denote the graph with adjacency matrix $W^t$ where $W$ is the lazy random walk matrix of $G$.
Then we have
\[
\phi(H)
\ge \frac1{20}(1 - (1 - \frac{\phi(G)}2)^{\sqrt t}).
\]
\end{theorem}

The following corollary is a generalization of Corollary 12 of \cite{powers}, which shows that general cases of improved Cheeger's inequality can be reduce to the cases where $\lambda_k$ is constant.

\begin{corollary}
\label{c:improved}
Suppose one could prove that $\phi(H) \le C \lambda_2(H)$ for some $C \ge 1 / 10$ whenever $\lambda_k(H) \ge 1 / 4$, then it implies that $\phi(G) \le 40 C \lambda_2(G) / \sqrt{\lambda_k(G)}$ for any $G$ and any $\lambda_k(G)$.
\end{corollary}

\begin{proof}
We assume that $\phi(G) \le \sqrt{\lambda_k}$, as otherwise, by Cheeger's inequality, $2 \lambda_2(G) \ge \phi(G)^2 \ge \phi(G) \sqrt{\lambda_k}$ and the statement is true.
Consider $H$ with adjacency matrix $W^{1 / \lambda_k(G)}$.
Then
\[
\lambda_k(H)
= 1 - (1 - \frac{\lambda_k(G)}2)^{1 / \lambda_k}
\ge 1 - \exp(-\frac12)
\ge 1 / 4.
\]
Therefore, if one could prove that $\phi(H) \le C \lambda_2(H)$, then
\[
C \lambda_2(H)
\ge \phi(H)
\ge \frac1{20} (1 - (1 - \frac{\phi(G)}2)^{\sqrt{1 / \lambda_k(G)}})
\ge \frac1{20} (1 - \exp(-\frac{\phi(G)}{2 \sqrt{\lambda_k(G)}}))
\ge \frac{\phi(G)}{80\sqrt{\lambda_k(G)}},
\]
where the second inequality is by Theorem~\ref{t:sqrtt}.
On the other hand,
\[
\lambda_2(H)
= 1 - (1 - \frac{\lambda_2(G)}2)^{1 / \lambda_k(G)}
\le \frac{\lambda_2(G)}{2 \lambda_k(G)},
\]
and the corollary follows by combining the two inequalities.
\end{proof}

Now we show the improved Cheeger's inequality in~\cite{improved-cheeger} follows from Corollary~\ref{c:improved} and Theorem~\ref{t:k-way}.
By the easy side of the higher order Cheeger's inequality, 
we have $\phi_k \ge \lambda_k/2$.
Hence, for any graph $G$ with $\lambda_k \ge 1 / 4$, we have $\phi_k \ge 1 / 8$ and Theorem~\ref{t:k-way} gives $\phi(G) = O(k \lambda_2(G))$.
Therefore, we can apply Corollary~\ref{c:improved} (with $C = O(k)$) and conclude that $\phi(G) = O(k \lambda_2(G) / \sqrt{\lambda_k(G)})$ is true for any graph $G$ and any $\lambda_k$,
and the improved Cheeger's inequality in~\cite{improved-cheeger} follows.

\end{document}